\documentclass[letterpaper,11pt]{article}

\usepackage[utf8]{inputenc}
\usepackage[letterpaper,hmargin=1in,vmargin=1in,footskip=36pt]{geometry}
\usepackage{mathptmx} 
\DeclareMathAlphabet\mathcal{OMS}{cmsy}{m}{n}
\SetMathAlphabet\mathcal{bold}{OMS}{cmsy}{b}{n}

\usepackage[font=footnotesize,labelfont=bf]{caption}

\usepackage{xspace}
\usepackage{cite}
\usepackage{amsmath, amssymb, amsthm}
\usepackage{paralist}
\usepackage{enumerate}
\usepackage{mathtools}
\usepackage[ruled,vlined]{algorithm2e}
\usepackage{tikz}
\usetikzlibrary{chains,arrows,decorations.pathreplacing}
\usepackage{qtree}
\def\ve#1{\mathchoice{\mbox{\boldmath$\displaystyle\bf#1$}}
	{\mbox{\boldmath$\textstyle\bf#1$}}
	{\mbox{\boldmath$\scriptstyle\bf#1$}}
	{\mbox{\boldmath$\scriptscriptstyle\bf#1$}}}

\newcommand{\G}{\ensuremath{\mathcal{G}}}

\newcommand\veb{{\ve b}}
\newcommand\vecc{{\ve c}}

\newcommand\veg{{\ve g}}
\newcommand\veG{{\ve G}}
\newcommand\veh{{\ve h}}
\newcommand\vel{{\ve l}}

\newcommand\veq{{\ve q}}

\newcommand\veu{{\ve u}}

\newcommand\vex{{\ve x}}
\newcommand\vey{{\ve y}}
\newcommand\vez{{\ve z}}

\newcommand\veQ{{\ve Q}}

\usepackage[textsize=tiny,textwidth=2cm,shadow,color=green]{todonotes}

\newtheorem{theorem}{Theorem}
\newtheorem{claim}{Claim}

\newtheorem{lemma}{Lemma}
\newtheorem{definition}{Definition}
\newtheorem{observation}{Observation}

\title{Covering a tree with rooted subtrees}

\author{
Lin Chen\thanks{Department of Computer Science, University of Houston. Email:
  \texttt{chenlin198662@gmail.com}.}
 \and Daniel Marx \thanks{Institute for Computer Science and Control,
 	Hungarian Academy of Sciences (MTA SZTAKI). Email:
   \texttt{dmarx@cs.bme.hu}.} 
}
\date{\today}

\begin{document}
\maketitle

\thispagestyle{empty}

\begin{abstract}
We consider the multiple traveling salesman problem on a weighted tree. In this problem there are $m$ salesmen located at the root initially. Each of them will visit a subset of vertices and return to the root. The goal is to assign a tour to every salesman such that every vertex is visited and the longest tour among all salesmen is minimized. The problem is equivalent to the subtree cover problem, in which we cover a tree with rooted subtrees such that the weight of the maximum weighted subtree is minimized. The classical machine scheduling problem can be viewed as a special case of our problem when the given tree is a star. We observe that, the problem remains NP-hard even if tree height and edge weight are constant, and present an FPT algorithm for this problem parameterized by the largest tour length. To achieve the FPT algorithm, we show a more general result. We prove that, integer linear programming that has a tree-fold structure is in FPT, which extends the FPT result for the $n$-fold integer programming by Hemmecke, Onn and Romanchuk~\cite{hemmecke2013n}.

\end{abstract}
\vspace{2mm}
\hspace{3.5mm}\textbf{Keywords:} Fixed Parameter Tractable; Integer Programming; Scheduling

\clearpage
\setcounter{page}{1}


\section{Introduction}

We consider the multiple traveling salesmen problem on a given tree $T=(V,E)$. In this problem there is a root $r\in V$ where all the $m$ salesmen are initially located. There is a weight $w_e\in \mathbb{Z}_+$ associated with each edge $e\in E$, which is the time consumed by a salesman if he passes this edge. Each salesman starts at $r$, travels a subset of the vertices and returns to $r$. The goal is to determine the tours traveled by each salesman such that every vertex is visited by some salesman, and the makespan, i.e., the time when the last salesman returns to $r$, is minimized.

We observe that the tour of every salesman is actually a subtree rooted at $r$, and the total traveling time of each salesman is exactly twice the total weight of edges in the subtree. Therefore the problem is equivalent as the minmax subtree cover problem, where we aim to find $m$ subtrees $T_i=(V(T_i),E(T_i))$ for $1\le i\le m$ such that $r\in V(T_i)$, $V=\cup_i V(T_i)$ and $\max_i w(T_i)$ is minimized, where $w(T_i)=\sum_{e\in E(T_i)} w_e$. We call $w(T_i)$ as the weight of the subtree $T_i$ and $\max_i w(T_i)$ the {\em makespan}.

The subtree cover problem is a fundamental problem in computer science and has received many studies in the literature. Indeed, when the given graph is a star, the problem is equivalent to the identical machine scheduling problem $P||C_{max}$, where the goal is to assign a set of jobs of processing times $w_1,w_2,\cdots,w_n$ onto $m$ identical parallel machines such the largest load among machines is minimized. We may view each job as an edge of weight $w_j$ in a star graph, whereas $P||C_{max}$ falls exactly into the problem of covering a star with $m$ stars. In 2013, Mnich and Wiese~\cite{mnich2015scheduling} provided an FPT (fixed parameter tractable) algorithm parameterized by the largest job processing time $w_{max}=\max\{w_j|1\le j\le n\}$.

The problem becomes much more complicated when the given graph is a tree. There exist some approximation algorithms for the problme, e.g., Xu et al.~\cite{xu2013exact} showed that there exists an FPTAS when the number of subtrees, $m$, is a constant. However, we are not aware of a paramerized algorithm for this problem. 

\vspace{1mm}
\noindent\textbf{Our contribution.} Our main contribution is to show that the subtree cover problem admits a fixed parameter tractable (FPT) algorithm (parameterized by the makespan). More precisely, we prove the following theorem.

\begin{theorem}\label{thm:main-fpt}
	For some computable function $f$, there exists an FPT algorithm of running time $f(B)m^{4}$ for determining whether there exists a feasible solution for the subtree cover problem of makespan $B$.
\end{theorem}

We remark that, despite the fact that the special case of covering a star admits an FPT algorithm parameterized by the largest edge weight, we show in this paper that the subtree cover problem remains NP-hard even if the tree is of height 2 and every edge has a unit weight. Therefore, we restrict our attention to the larger parameter $B$.

Indeed, our FPT algorithm relies on an FPT algorithm for a more general integer programming problem, which extends the existing FPT algorithm for the $n$-fold integer programming~\cite{hemmecke2013n}. We consider the following integer programming:
\begin{eqnarray}\label{eq:tree-fold}
	\min\{\vecc^T\vex: A\vex=\veb, \vel\le \vex\le \veu, \vex\in \mathbb{Z}^{nt}\},
\end{eqnarray}

In the $n$-fold integer programming, the matrix $A$ consists of small matrices $A_1$ and $A_2$ as follows (Here $A_1$ is an $s_1\times t$-matrix and $A_2$ is an $s_2\times t$-matrix).
\begin{eqnarray*}
	A=
	\begin{bmatrix}
		A_{1} & A_{1} & \dots  & A_{1} \\
		A_{2} & 0     & \dots    &  0 \\
		0     &   A_2  & \dots     & 0  \\
		\vdots& \vdots& \ddots   &\vdots\\
		0     &   0   & \dots  &  A_2  \\
	\end{bmatrix}
\end{eqnarray*}
More precisely, the matrix $A$ consists of one row of $(A_1,A_1,\cdots,A_1)$ and a submatrix with $A_2$ being at the main diagonal. We remark that throughout this paper $0$s that appear in a matrix refer to a submatrix consisting of the natural number $0$.

The $n$-fold integer programming has received many studies in the literature. Indeed, the natural ILP formulation of the scheduling and bin packing problem falls into an $n$-fold integer programming, as is observed by Knop and Kouteck{\'y}~\cite{knop2016scheduling}. In 2013, Hemmecke, Onn and Romanchuk presented an FPT algorithm for $n$-fold integer programming with the running time of $f(s_1,s_2,||A||_{\infty})n^3L$ where $f$ is some computable function, $||A||_{\infty}$ is the largest absolute value among all entries of $A$ and $L$ is the encoding length of the problem. This algorithm implies an FPT algorithm parameterized by the largest job processing time for $P||C_{max}$ and many other scheduling problems~\cite{knop2016scheduling}. We further extend their result by considering a broader class of integer programming, namely tree-fold integer programming as we describe as follows.

The structure of an $n$-fold matrix could be viewed as a star with the root representing the row of $(A_1,A_1,\cdots,A_1)$ and each leaf representing one of the rows $(0,\cdots,0,A_2,0,\cdots,0)$. More precisely, we can view each row $i$ as a vertex $i$ such that vertex $i$ is a parent of vertex $j$ if row $i$ dominates row $j$, where by saying row $i$ dominates row $j$, we mean row $j$ is more "sparse" than row $i$ as a vector, i.e., if the $k$-th coordinate of row $j$ is non-zero, then the $k$-th coordinate of row $i$ is also non-zero. Using this interpretation, we can generalize an $n$-fold matrix to a tree-fold matrix. The following is an example.

\begin{eqnarray*}
	A=
	\begin{bmatrix}

		A_1   & A_{1} & A_{1} & A_1 & A_1 & A_1  & A_{1} & A_1 & A_1 & A_1 & A_1  & A_1 \\
		A_{2} & A_{2} & A_{2} & A_2 & A_2 & A_2  & A_{2} & A_2   & 0   &  0  & 0  & 0  \\
		0     & 0     &   0   & 0  & 0  & 0 & 0 & 0  & A_2 & A_2 & A_2   & A_2 \\
		A_{3} & A_{3} &   A_3   & 0 & 0  & 0     & 0   & 0   &  0  & 0 & 0  & 0  \\
		0     & 0     &   0 & A_3  & A_3   & 0   & 0   &  0  & 0 & 0  & 0 & 0 \\
		0     & 0     &   0   & 0  & 0     & A_3 & A_3 &  A_3  & 0  & 0 & 0 & 0 \\
	0     & 0     &   0   & 0  & 0     & 0 & 0 &  0  & A_3  & A_3 & A_3 & A_3 \\
	A_4     & 0     &   0   & 0  & 0     & 0 & 0 &  0  & 0  & 0 & 0 & 0 \\
	0    & A_4     &   0   & 0  & 0     & 0 & 0 &  0  & 0  & 0 & 0 & 0 \\
	0    & 0     &   A_4   & 0  & 0     & 0 & 0 &  0  & 0  & 0 & 0 & 0 \\
		0    & 0     &   0   & A_4  & 0     & 0 & 0 &  0  & 0  & 0 & 0 & 0 \\
			0    & 0     &   0   & 0  & A_4     & 0 & 0 &  0  & 0  & 0 & 0 & 0 \\
				0    & 0     &   0   & 0  & 0     & A_4 & 0 &  0  & 0  & 0 & 0 & 0 \\
		0    & 0     &   0   & 0  & 0    & 0 & A_4 &  0  & 0  & 0 & 0 & 0 \\		
		0    & 0     &   0   & 0  & 0    & 0 & 0 &  A_4  & 0  & 0 & 0 & 0 \\	
			0    & 0     &   0   & 0  & 0    & 0 & 0 &  0  & A_4  & 0 & 0 & 0 \\	
		0    & 0     &   0   & 0  & 0    & 0 & 0 &  0  & 0  & A_4 & 0 & 0 \\		
		0    & 0     &   0   & 0  & 0    & 0 & 0 &  0  & 0  & 0 & A_4 & 0 \\					
			0    & 0     &   0   & 0  & 0     & 0 & 0 &  0  & 0  & 0 & 0 & A_4 \\
	\end{bmatrix}
\end{eqnarray*}
A tree-representation of the matrix above is:

\begin{center}\label{fig:tree}
	\includegraphics[scale=0.4]{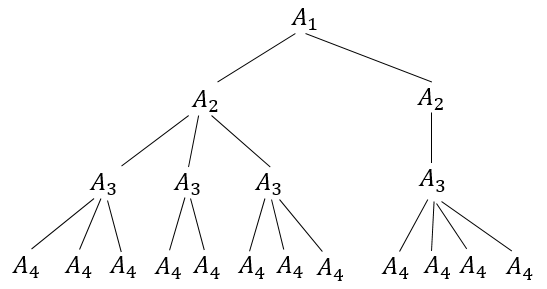}
\end{center}


In general, a tree-fold matrix $A$ consists of $n$ copies of small matrices $A_1$, $A_2$, $\cdots$, $A_{\tau}$ with $A_i$ being an $s_i\times t$-matrix. Every row consists of $0$'s and some $A_i$'s in the form of $(0,\cdots,0,A_i, A_i,\cdots,A_i,0,\cdots,0)$ (i.e., $A_i$ appears consecutively). Every column consists of $0$'s and exactly one copy of each $A_i$. Furthermore, if we call a row containing $A_i$ as an $A_i$-row, then any $A_i$-row is dominated by some $A_{i-1}$-row, that is, if at a certain row $A_i$ appears consecutively from column $\ell$ to column $k$, then there exists some $A_{i-1}$-row such that $A_{i-1}$ appears consecutively from $\ell'$ to $k'$ such that $\ell'\le \ell<k\le k'$. Representing the matrix as a tree, every row is represented as a vertex and the vertex corresponding to each $A_{i-1}$-row will be the parent of the vertex corresponding to $A_i$-row it dominates.

To facilitate the analysis, we further require that the $A_1$-row contains no $0$ and every $A_\tau$-row contains exactly one copy of $A_\tau$, that is, all rows containing $A_{\tau}$ form a sub-matrix with $A_\tau$ being at the diagonal. Note that this assumption causes no loss of generality: If it is not the case, we can always add a set of dummy constraints: $0\cdot \vex=0$, whereas $A_1$ and $A_\tau$ become a $1\times t$-dummy matrix consisting of $0$.


We define ILP~(\ref{eq:tree-fold}) with $A$ being a tree-fold matrix as a tree-fold integer programming and establish the following FPT result.

\begin{theorem}\label{thm:main-ILP}
	For some computable function $f$, there exists an FPT algorithm of running time $f(t,s_1,s_2,\cdots,s_\tau,||A||_{\infty})n^3L$ for a tree-fold integer programming, where $||A||_{\infty}$ is the largest absolute value among all entries of $A$, and $L$ is the length of the binary encoding of the vector $(\vecc, \veb,\vel,\veu)$.
\end{theorem}

Note that $||A||_{\infty}=\max_j\{||A_j||_{\infty}\}$, thus the FPT term $f(t,s_1,s_2,\cdots,s_\tau,||A||_{\infty})$ only depends on the small matrices and does not rely on the structure of $A$. We also remark that, by introducing slack variables for inequalities, our theorem also holds for the integer programming:
	$\min\{\vecc^T\vex: A\vex\le \veb, \vel\le \vex\le \veu, \vex\in \mathbb{Z}^{nt}\}$.
	
	\vspace{2mm}
\noindent\textbf{Related work.} As we have mentioned, the problem of covering a star with stars is exactly the identical machine scheduling problem $P||C_{max}$. Approximation schemes are studied in a series of prior papers, see, e.g.,~\cite{chen2014optimality, sahni1976algorithms, hochbaum1987using, JansenKV16,jansen2010technical}. In terms of FPT algorithms, Mnich and Wiese~\cite{mnich2015scheduling} showed that $P||C_{max}$ is FPT parameterized by the largest job processing time (edge weight). Very recently, Knop and Kouteck{\'y}~\cite{knop2016scheduling} observes the relationship between the scheduling problem and $n$-fold integer programming in terms of FPT algorithms. Indeed, they show that a variety of scheduling problems, including $P||C_{max}$, could be formulated as an $n$-fold integer programming. Applying the FPT algorithm for $n$-fold integer programming by Hemmecke, Onn and Romanchuk~\cite{hemmecke2013n}, an FPT algorithm for $P||C_{max}$ follows. It is worth mentioning that parameterized studies for integer programming that has a sparse structure have received much attention in the literature, e.g.,~\cite{jansen2015structural,kratsch2016polynomial}.

Covering a tree with subtrees is much more complicated. In 2013, Xu et al.~\cite{xu2013exact} showed that if the number of subtrees, $m$, is a constant, then the problem admits a pseudo-polynomial time exact algorithm and an FPTAS. We are not aware of FPT algorithms for this problem. 



\section{The FPT algorithm}\label{sec:fpt}
In this section, we show that the subtree cover problem is FPT parameterized by the makespan. Towards this, we formulate the problem as an ILP. We observe that the ILP we establish has a special structure, which generalizes the $n$-fold integer programming studied in the literature. We call it as a tree-fold integer programming. Indeed, when the input tree is a star, the tree-fold integer program we formulate becomes an $n$-fold integer program. We extend the FPT algorithm for the $n$-fold integer programming to derive an FPT algorithm for the tree-fold integer programming, which implies an FPT algorithm for the subtree cover problem. This result may be of separate interest.

Recall that when the given graph is a star, the subtree cover problem becomes FPT parameterized by the largest edge weight $w_{max}=\max_j\{w_j|1\le j\le n\}$~\cite{mnich2015scheduling}. However, this is no longer true even if the given graph is a tree of height $2$, as is implied by the following theorem.

\begin{theorem}\label{thm:hardness}
	The subtree cover problem remains NP-hard even if the given tree is of height $2$ and every edge has unit weight.
\end{theorem}

The above hardness result excludes FPT algorithms parameterized by edge weight and tree height, and therefore we restrict our attention to makespan. We will first show that a tree-fold integer programming can be solved in FPT time. Then we establish a configuration ILP for the subtree cover problem and prove that the ILP falls exactly into the category of tree-fold integer programming, and is thus solvable in FPT time.




\subsection{Tree-fold integer programming}
The goal of this and next subsection is to prove Theorem~\ref{thm:main-ILP}. Towards this, we first introduce some basic concepts and techniques which are crucial for our proof. Here we only give a very brief introduction and the reader may refer to Appendix~\ref{ap-sec:pre-tree} for details.

We consider the following integer programming with $A$ being a tree-fold matrix consisting of $n$ copies of $s_i\times t$-matrix $A_i$, where $i=1,2,\cdots,\tau$.
\begin{eqnarray}\label{eq:IP-main}
\min\{\vecc^T\vex: A\vex=\veb, \vel\le \vex\le \veu, \vex\in \mathbb{Z}^{nt}\},
\end{eqnarray}

Any vector $\vex\in \mathbb{Z}^{nt}$ can be written into $n$ "bricks" in the form of $(\vex^1,\vex^2,\cdots,\vex^n)$ where $\vex^i\in \mathbb{Z}^{t}$. Using the standard technique, we can prove that if we have an algorithm for a tree-fold ILP such that given a feasible initial solution, it can augment it into an optimal solution, then by using this algorithm as a subroutine we can also solve the tree-fold ILP without knowing the initial solution (see Appendix~\ref{ap-sec:initial-solution}). Therefore, it suffices to focus on the "augmenting" algorithm. It is easy to see that all the vectors that can be used to augment a feasible solution $\vex$ to $\vex+\veq$ should satisfy that $A\veq=0$. It is shown by Graver~\cite{graver1975foundations} that instead of considering all the $\veq\in Ker(A)$, it suffices to consider a subset $\mathcal{G}(A)$, which is called Graver basis. Hemmecke, Onn and Weismantel~\cite{hemmecke2011polynomial} proved that, starting from an arbitrary feasible solution $\vex_0$, the optimal solution $\vex^*$ could be achieved by iteratively applying the {\em best augmentation via Graver basis}, i.e., augmenting $\vex$ by using the best possible augmentation vector of the form $\gamma\veg$, where $\gamma\in\mathbb{Z}_+$ and $\veg\in \mathcal{G}(A)$. The total number of augmentation steps needed is bounded by $O(nL)$, where $L$ is the length of the binary encoding of the vector $(\vecc,\veb,\vel,\veu)$\footnote{It should be noted that the best augmentation via Graver basis needs not be the best augmentation (i.e., there may exist $\veq$ such that $\vex+\veq$ is better than any $\vex+\gamma\veg$.}. This statement remains true if, instead of choosing the best possible augmentation vector of the form $\gamma\veg$, say, $\gamma^*\veg^*$, we choose an augmentation vector $\veq$ which is at least as good as $\gamma^*\veg^*$ in every augmentation step. That is, if in each augmentation step we choose an augmentation vector $\veq$ such that $\vecc^T\veq\le \gamma^*\vecc^T\veg^*$, then the optimal solution $\vex^*$ could also be achieved after $O(nL)$ augmentation vectors~\cite{hemmecke2013n,de2013algebraic}. Notice that $\veq$ does not necessarily belong to $\G(A)$. Such an augmentation is called a Graver-best augmentation and such greedy algorithm is called {\em Graver-best augmentation algorithm}. 

As we have described above, the problem of solving a tree-fold integer programming reduces to the problem that, given a feasible solution, finding an augmentation vector that is at least as good as the best augmentation via Graver basis. Towards this, it is crucial to understand the structure of the Graver basis for $A$. The following lemma provides such structural information and is crucial to our algorithm.

\begin{lemma}\label{le:fitness}
	Let $A=T[A_1,A_2,\cdots,A_{\tau}]$. There exists some integer $\lambda=\lambda(A_1,A_2,\cdots, A_{\tau})$ that only depends on matrices $A_1$ $A_2$, $\cdots$, $A_{\tau}$, and
	$$H(A)=\{\veh\in \mathbb{Z}^t| \veh \text{ is the sum of at most $\lambda$ elements of } \mathcal{G}(A_{\tau})\},$$
	such that for any $\veg=(\veg^1,\veg^2,\cdots,\veg^n)\in\mathcal{G}(A)$
	we have $\sum_{i\in I}\veg^i\in H(A)$ for any $I\subseteq \{1,2,\cdots,n\}$.
\end{lemma}
Here $A=T[A_1,A_2,\cdots,A_{\tau}]$ means $A$ is a tree-fold matrix consisting of $A_1$, $\cdots$, $A_\tau$. Roughly speaking, Lemma~\ref{le:fitness} states that for any Graver basis element $\veg$ of the matrix $A$, although it is of a very high dimension, it is sparse, i.e., among the $n$ bricks $\veg^1,\veg^2,\cdots,\veg^n$, only an \lq\lq FPT\rq\rq\, number of them can be nonzero. This lemma extends the structural lemma for $n$-fold integer programming in~\cite{hemmecke2013n}, which can be viewed as the case when $\tau=2$. The proof of Lemma~\ref{le:fitness} is involved and is deferred to Appendix~\ref{ap-sec:fitness}.



\subsection{Dynamic programming in FPT time}
We provide a dynamic programming algorithm running in FPT algorithm for the tree-fold integer programming, and Theorem~\ref{thm:main-ILP} follows. Towards this, we let $\lambda=\lambda(A_1,A_2,\cdots,A_{\tau})$ and $H(A)$ be defined as in Lemma~\ref{le:fitness}.

Given a feasible solution $\vex$ of the integer programming~(\ref{eq:IP-main}), let $\gamma^*\in\mathbb{Z}_+$, $\veg^*\in \G(A)$ satisfy that $\gamma^*\veg^*$ is the best augmentation among Graver basis, i.e., the best possible augmentation vector of the form $\gamma\veg$ where $\gamma\in\mathbb{Z}_+$ and $\veg\in\G(A)$. The following lemma from~\cite{hemmecke2013n} allows us to guess $\gamma^*$ in $O(n)$ time:

\begin{lemma}[\cite{hemmecke2013n}]\label{lemma:cite}
	In $O(n)$ time we can compute a set of integers $\Gamma$ such that $\gamma^*\in\Gamma$ and $|\Gamma|\le n|H(A)|$.
\end{lemma}
The proof in~\cite{hemmecke2013n} is for the case when $\tau=2$, however, it works directly for the general tree-fold matrices. For the completeness of the paper we give the proof in Appendix~\ref{ap-sec:cite}.

In the following we give a dynamic programming algorithm such that given a feasible solution $\vex$ and any $\gamma\in\Gamma$, it finds out $\veh_{\gamma}\in H(A)$ that minimizes $\vecc^T(\vex+\gamma\veh_{\gamma})$, or equivalently, minimizes $\vecc^T\veh_{\gamma}$ subject to the constraints that $\vel^i\le \vex^i+\gamma\veh_{\gamma}^i\le \veu^i$ and $A\veh_{\gamma}=0$. With such an algorithm, we can run it for every $\gamma\in\Gamma$ and pick $\gamma'$ such that $\vecc^T(\vex+\gamma'\veh_{\gamma'})$ is minimal. By the definition of $H(A)$, $\gamma'\veh_{\gamma'}$ is at least as good as the best augmentation via Graver basis and is thus the Graver-best augmentation that we desire. 

The dynamic programming works in stages where in each stage it solves a subproblem. To define the subproblem, we define a matrix $\bar{A}$ as follows. Consider any small matrix $A_i$ and all the rows in $A$ that contain $A_i$. Suppose $A_i$ appears consecutively in these rows from column $1=d_0^i$ to column $d_1^i$, from column $d_1^i+1$ to column $d_2^i$, $\cdots$, from column $d_{k-1}^i$ to column $d_k^i=n$. 
We define $\bar{A}$ where each row of $\bar{A}$ is the summation of some rows in $A$. More precisely, $\bar{A}$ contains the same number of rows as $A$. If in the $\ell$-th row of $A$ some small matrix $A_i$ appears consecutively from column $d_j^i$ to column $d_{j+1}^i$, then in the $\ell$-th row of $\bar{A}$ the small matrix $A_i$ appears consecutively from $1$ to $d_{j+1}^i$, that is, we construct $\bar{A}$ by extending the sequence of $A_i$ in each row of $A$ to column $1$. It is obvious that $A\veh=0$ if and only if $\bar{A}\veh=0$.

Let $\bar{A}[1],\bar{A}[2],\cdots$ be all the rows in $\bar{A}$. Let $ED_k$ be the set of rows $\bar{A}[\ell]$ where only the first $k$ columns are non-zero. Obviously $ED_k\subseteq ED_{k+1}$. Let $H_{max}=\max_{\veg\in\G(A)}||\veg||_1$ and $Q_h= \{z\in \mathbb{Z}^{s_h}: ||z||_{1}\le ||A_h||_1\cdot H_{max}\}$. According to Lemma~\ref{le:fitness}, $H_{max}$, and hence $||z||_1$ for any $z\in Q_h$, is only dependent on the submatrices $A_1,A_2,\cdots,A_{\tau}$. 
We define subproblem-$k$ as follows: 

For every $z_h\in Q_h$ where $1\le h\le \tau$, find some $\bar{\veh}_{\gamma}$ such that 
\begin{itemize}
	\item $\bar{\veh}_{\gamma}^i=0$ for $i>k$, that is, only the first $k$ bricks can be non-zero.
	\item $\bar{\veh}_{\gamma}\in H(A)$.
	\item $\vel^i\le \vex^i+\gamma\bar{\veh}_{\gamma}^i\le \ve u^i$ for $1\le i\le k$.
	\item $\bar{A}[\ell]\cdot \bar{\veh}_\gamma=0$ for any $\bar{A}[\ell]\in ED_k$.
	\item $\sum_{i}A_hx^i=z_h$,
	\item $\vecc^T\bar{\veh}_{\gamma}$ is minimized.
\end{itemize} 

It is easy to see that the optimal solution for the subproblem-$(k+1)$ can be constructed by extending the optimal solution for the subproblem-$k$ by one brick, and such a brick belongs to $H(A)$. Therefore, the optimal solution for subproblem-$n$ can be found in $O(n)$ time, where the big-O hides a coefficient that only depends on $A_1,A_2,\cdots,A_{\tau}$.

\vspace{1mm}
\noindent\textbf{The overall running time.} We have shown in this subsection that the dynamic programming algorithm can find out a Graver best augmentation in $O(n^2)$ time (ignoring all the FPT-terms). By~\cite{hemmecke2011polynomial} the number of Graver best augmentations needed is $O(nL)$ where $L$ is the encoding length of the integer programming, therefore tree-fold integer programming can be solved in $O(n^3L)$ time, and Theorem~\ref{thm:main-ILP} is proved (if a feasible initial solution is given).





\subsection{Subtree cover--integer programming formulation}
The goal of this subsection is to derive an ILP formulation of the subtree cover problem which falls into the category of tree-fold integer programming. Given this result, applying Theorem~\ref{thm:main-ILP}, Theorem~\ref{thm:main-fpt} is proved.

For ease of description, we let the root $r=v_1$. We define the unweighted distance between two vertices as the length the path connecting them in the same tree with all edge weights as $1$. The {\em depth} of any vertex $v_s$ is the unweighted distance of $v_s$ to $v_1$. 

\vspace{1mm}
\noindent\textbf{Preprocessing.} We consider the decision version of the problem which asks whether there exists a subtree cover of makespan $B$. We assume without loss of generality that the height of the tree, $h(T)$, is at most $B$, since otherwise we can conclude directly that there is no feasible solution of makespan at most $B$. For ease of presentation, we modify the problem in the following way. For any leaf whose depth is $h< h(T)$, we append a path to it which consists of $h(T)-h$ dummy vertices and $h(T)-h$ dummy edges of $0$ weight. By doing so every leaf of $T$ has a depth of $h(T)$. Next, we direct all the edges towards the root and move the weight of each edge to its source vertex. Specifically, the weight of the root is $0$. Now the weight of any subtree is simply the total weight of its vertices. For simplicity, we still denote the modified tree as $T$ and denote by $n$ the number of its vertices.

\vspace{1mm}
\noindent\textbf{Configurations.} We define configurations.
Any tree with at most $O(B^2)$ vertices whose weight is bounded by $B$ can be encoded via an $O(B^2)$-vector as follows: We index all vertices from $1$ to $O(B^2)$. For every vertex, we store its weight and its parent. We call such an $O(B^2)$-vector as a {\em configuration} and have the following simple observation.

\begin{observation}
	There are at most $\mu=B^{O(B^2)}$ different kinds of configurations.
\end{observation}

We index configurations arbitrarily as $CF_1,CF_2,\cdots,CF_{\mu}$ and denote by $|CF_j|$ the number of vertices in $CF_j$. Given an arbitrary configuration $CF_j$, we use $(CF_j,k)$ to denote its vertex of index $k\in \{1,2,\cdots,|CF_j|\}$. $k$ is also called the {\em location} of this vertex. Let $\zeta=O(B^2)$ be the maximal number of vertices among all the configurations. A pair $(CF_j,k)$ with $|CF_j|<k\le \zeta$ is called invalid. For simplicity, $1$ is always the index (location) of the root for every $CF_j$.

Given a configuration $CF_j$, we define a function $f_{j}$ which maps a vertex of location $k$ to the location of its parent (it shall be noted that here the function $f_j$ has nothing to do with the function $f$ in Theorem~\ref{thm:main-ILP}).

Now we revisit the subtree cover problem using the notion of configurations. Consider an arbitrary subtree of $T$ rooted at $r=v_1$ whose weight is at most $B$. We first observe that there are at most $O(B^2)$ vertices in the subtree. To see why, we can first consider a subtree of weight at most $B$ in the original tree before preprocessing. Since every vertex, except the root, has non-zero weight, the number of vertices is bounded by $B+1$. As the preprocessing procedure will append at most $h(T)\le B$ vertices below a vertex, the total number of vertices is thus bounded by $O(B^2)$. Hence, any subtree of weight at most $B$ can be mapped to a configuration. Furthermore, any feasible solution can be interpreted as $m$ subtrees that can be mapped to $m$ configurations. Using this idea, we now establish an ILP formulation of the problem. 

We define an integral variable $x_{i,(CF_j,k)}$ for every vertex $v_i$ and every pair $(CF_j,k)$. For $h\in\mathbb{Z}_+$, $x_{i,(CF_j,k)}=h$ implies that there are $h$ subtrees in the solution which contain $v_i$, and furthermore, each of them can be mapped to the configuration $CF_j$ such that $v_i$ is mapped to the location $k$ vertex in $CF_j$. 

Obviously, $v_i$ can not be mapped to an arbitrary vertex in $CF_j$.
We say a vertex $v_i$ is consistent with the pair $(CF_j,k)$, if both of the following conditions are true:
\begin{itemize}
	\item the depth of $v_i$ in $T$ is the same as the depth of the location $k$ vertex in $CF_j$;
	\item the weight of $v_i$ in $T$ is the same as the weight of the location $k$ vertex in $CF_j$.
\end{itemize}
Otherwise, we say they are inconsistent. 

Let $CH(v_i)$ be the set of children of $v_i$, $LF$ be the set of leaves. We establish the following $ILP(T)$ for the subtree cover problem:
		\begin{subequations}\label{eq:config-ILP}
			\begin{align*}
			&\quad\quad\min \sum_{j=1}^{\mu} x_{1,(CF_j,1)}&\\
			&(I)\quad \sum_{s:v_s\in CH(v_i)} x_{s,(CF_j,k)}=x_{i,(CF_j,f_{j}(k))},\quad &\forall\,1\le i\le n, 1\le j\le \mu, 1\le k\le \zeta\\
			&(II)\quad \sum_{j=1}^{\mu}\sum_{k=1}^{\zeta} x_{i,(CF_j,k)}=1, \quad& \forall\, v_i\in LF\\
			& (III)\quad x_{i,(CF_j,k)}=0, \quad& \textrm{if $v_i$ and $(CF_j,k)$ are inconsistent, or }|CF_j|<k\le \zeta\\ 
			& (IV)\quad  x_{i,(CF_j,k)}\in \mathbb{Z}_{\ge0}, \quad& 1\le i\le n, 1\le j\le \mu, 1\le k\le \zeta  
			\end{align*}
		\end{subequations}
Constraint $(II)$ ensures that every leaf is contained in one of the subtrees. Constraints $(III)$ and $(IV)$ are straightforward. We now explain constraint $(I)$. Consider any feasible solution and let $v_i$ be an arbitrary vertex. Let $v_s$ be any child of $v_i$. If $v_s$ is mapped to the vertex of location $k$ in $CF_j$, then $v_i$ must be mapped to the vertex of location $f_j(k)$ in $CF_j$. Therefore, if we consider the total number of configuration $CF_j$ where a child of $v_i$ is mapped to its vertex of location $k$, this should be equal to the number of configuration $CF_j$ where $v_i$ is mapped to its vertex of location $f_j(k)$. This is essentially what constraint $(I)$ implies.

The following two lemmas ensures that the $ILP(T)$ we have derived indeed solves the subtree cover problem. One direction (Lemma~\ref{lemma:solution-to-ILP}) is staightforward, yet the other direction is a bit involved and the reader is referred to Appendix~\ref{ap-sec:ILP-to-solution} for details.

\begin{lemma}\label{lemma:solution-to-ILP}
	If there exists a feasible solution of the scheduling problem with makespan at most $B$, then there exists a feasible solution of the ILP with the objective value at most $m$.
\end{lemma}

\begin{lemma}\label{lemma:ILP-to-solution}
	If there exists a feasible solution of the ILP with the objective value at most $m$, then there exists a feasible solution of the subtree cover problem with makespan at most $B$.
\end{lemma}

Still, $ILP(T)$ is similar but not exactly the same as a tree-fold integer programming. We need to tune the ILP a bit. The tuning is essentially by replacing some of the variables with the equation in $(I)$ it satisfies, i.e., we will remove some of the variables. See Appendix~\ref{ap-sec:tuning} for details. Once transformed into a tree-fold integer programming, Theorem~\ref{thm:main-ILP} can be applied and Theorem~\ref{thm:main-fpt} is proved.

\section{Conclusion}
We consider the subtree cover problem in this paper and provide an FPT algorithm parameterized by the makespan. 
Our FPT algorithm follows from a more general FPT result on the tree-fold integer programming, which extends the existing FPT algorithm on the $n$-fold integer programming. The running times of the FPT algorithms is huge and is only of theoretical interest. Another important open problem is whether we can derive FPT algorithm for integer programming with the matrix $A$ that has an even more general structure. It is also interesting to consider approximation schemes for the subtree cover problem.

\appendix

\section{Proofs Omitted in Section~\ref{sec:fpt}}

\subsection{Proof of Theorem~\ref{thm:hardness}}
\begin{proof}[Proof of Theorem~\ref{thm:hardness}]
	We reduce from $3$-partition. In the $3$-partition problem, given is a set of $3n$ integers $a_1,a_2,\cdots,a_{3n}$ with $B/4<a_j<B/2$, $\sum_j a_j=3nB$ where $B=n^{O(1)}$. The goal is to determine whether we can partition the $3n$ integers of $n$ subsets $D_1,D_2,\cdots,D_n$, each of size $3$, such that $\sum_{a_j\in D_i} a_j=B$ for every $1\le i\le n$.
	
	We construct a subtree cover instance as follows. There is a root $r$. The root has $3n$ children $v_1,v_2,\cdots,v_{3n}$. Each $v_j$ further has $a_j$ children. We let the weight of every edge be $1$. 	
	
	We show that the constructed subtree cover instance can be covered by $n$ subtrees of makespan $B+3$ if and only if the given 3-partition instance admits a feasible partition. 
	
	Suppose the 3-partition instance admits a feasible partition, then each subtree consists of the root, $\{v_j|a_j\in S_i\}$ and their children. It is easy to verify that the weight of each subtree is exactly $B+3$. 
	
	Suppose the subtree cover instance admits a solution of makespan $B+3$. Since all edge weights sum up to $nB+3n$, we know each subtree consists of exactly $B+3$ edges, and each edge appears in one subtree. Therefore, if a subtree contains a vertex $v_j$, it must contain all the children of $v_j$. As $v_j$ has $B/4<a_j<B/2$ children, it is easy to see that each subtree contains exactly 3 children of the root, implying readily a solution for the 3-partition instance.
\end{proof}

\subsection{Preliminaries for Tree-fold Integer Programming}\label{ap-sec:pre-tree}
We provide a brief introduction to the notions needed for solving a general integer programming. We refer the readers to a nice book~\cite{de2013algebraic} for details. 

We define {\em Graver basis}, which was introduced in~\cite{graver1975foundations} by Graver and is crucial for our algorithm. 

We define a partial order $\sqsubseteq$ in $\mathbb{R}^n$ in the following way:
$$\textrm{For any } \vex,\vey\in \mathbb{R}^n,\,\, \vex\sqsubseteq \vey\text{ if and only if for every } 1\le i\le n, |x_i|\le |y_i| \text{ and } x_i\cdot y_i\ge 0. $$ 
Roughly speaking, $\vex\sqsubseteq \vey$ implies that $\vex$ and $\vey$ lie in the same orthant, and $\vex$ is \lq\lq closer\rq\rq\, to the origin $0$ than $\vey$. The partial order $\sqsubseteq$, when restricted to $\mathbb{R}^n_+$, coincides with the classical coordinate-wise partial order $\le$. 

Given any subset $X\subseteq \mathbb{R}^n$, we say $\vex$ is an $\sqsubseteq$-minimal element of $X$ if $\vex\in X$ and there does not exist $\vey\in X$, $\vey\neq \vex$ such that $\vey\sqsubseteq \vex$. 

According to Gordan's Lemma, for any subset $Z\subseteq \mathbb{Z}^n$, the number of $\sqsubseteq$-minimal elements in $Z$ is finite. Indeed, this fact is known as Dickson's Lemma for the coordinate-wise partial order $\preceq$. 


\begin{definition}
	The Graver basis of an integer $m\times n$ matrix $A$ is the
	finite set $\mathcal{G}(A)\subseteq \mathbb{Z}^n$ which consists of all the $\sqsubseteq$-minimal elements of $ker_{\mathbb{Z}^n}(A)=\{\vex\in \mathbb{Z}^n| A\vex=0,\vex\neq 0\}$.
\end{definition}

The Graver basis $\mathcal{G}(A)$ is only dependent on $A$. Let $||B||_{\infty}$ be the largest absolute value over all entries. For any $\veg\in\mathcal{G}(A)$, we have the following rough estimation for some constant $c_1,c_2$~\cite{onn2010nonlinear}:
$$|\mathcal{G}(A)|\le (c_1||A||_{\infty})^{mn}\quad\text{ and } ||g||_{\infty}\le (c_2||A||_{\infty})^{mn}.$$

The Graver basis has the following {\em positive sum property}: for every $\vez\in ker_{\mathbb{Z}^n}(A)$, there exist a subset $U\subseteq \mathcal{G}(A) $ such that for every $\veg_i\in U$, $\veg_i\sqsubseteq \vez$, and furthermore, $\vez=\sum_{\veg_i\in U}\alpha_i\veg_i$ for some $\alpha_i\in \mathbb{Z}_+$. See~\cite{onn2010nonlinear,de2013algebraic} for details.

Given is an integer programming of the following form:
\begin{eqnarray}\label{eq:IP}
\min\{\vecc^T\vex| A\vex=\veb, \vel\le\vex\le\veu, \vex\in\mathbb{Z}^n\}.
\end{eqnarray}
Let $\vex$ be an arbitrary feasible solution of~(\ref{eq:IP}). We say $\veq$ is an {\em augmentation vector} for $\vex$ if $\vex+\veq$ is a feasible solution of~(\ref{eq:IP}) that has an objective value strictly better than $\vex$, i.e., $\vecc^T(\vex+\veq)<\vecc^T\vex$. Therefore, $A\veq=0$ and $\vecc^T\veq<0$.

It is shown by Graver~\cite{graver1975foundations} that $\vex^*$ is an optimal solution of~(\ref{eq:IP}) if and only if there does not exist $\veg\in \mathcal{G}(A)$ which is an augmentation vector for $\vex^*$. Later on, Hemmecke, Onn and Weismantel~\cite{hemmecke2011polynomial} proved that, starting from an arbitrary feasible solution $x_0$ for~(\ref{eq:IP}), the optimal solution $\vex^*$ could be achieved by iteratively applying the {\em best augmentation via Graver basis}, i.e., augmenting $\vex$ by using the best possible augmentation vector of the form $\gamma\veg$, where $\gamma\in\mathbb{Z}_+$ and $\veg\in \mathcal{G}(A)$. The total number of augmentation vectors needed is bounded by $O(nL)$, where $L$ is the length of the binary encoding of the vector $(\vecc,\veb,\vel,\veu)$ (There may exist an augmentation vector which is better than any Graver basis, however, the result of~\cite{hemmecke2011polynomial} allows us to restrict our attention to Graver basis). This statement remains true if, instead of choosing the best possible augmentation vector of the form $\gamma\veg$, say, $\gamma^*\veg^*$, we choose an augmentation vector $\veq$ which is at least as good as $\gamma^*\veg^*$. That is, if in each augmentation vector we choose an augmentation vector $\veq$ such that $\vecc^T\veq\le \gamma^*\vecc^T\veg^*$, the optimal solution $\vex^*$ could also be achieved after $O(nL)$ augmentation vectors~\cite{hemmecke2013n,de2013algebraic}. Notice that $\veq$ does not necessarily belong to $\G(A)$. Such greedy algorithm is called {\em Graver-best augmentation algorithm}.  

The results by Hemmecke et al.~\cite{hemmecke2013n,de2013algebraic} imply that, to design a polynomial time algorithm for~(\ref{eq:IP}), it suffices to handle the following two problems:
\begin{itemize}
	\item[a.] finding a feasible initial solution for~(\ref{eq:IP}) in polynomial time;
	\item[b.] finding a Graver-best augmentation algorithm that runs in polynomial time.
\end{itemize}
In Subsection~\ref{ap-sec:initial-solution} we show in detail how to find a feasible initial solution for~(\ref{eq:IP}) in polynomial time. Roughly speaking this could be handled by establishing another ILP with a trivial initial feasible solution and finding its optimal solution.

We focus on problem [b]. A natural algorithm is that, given the current feasible solution $\vex$, for every $\veg\in\mathcal{G}(A)$, we find integer $\gamma_{\veg}\in\mathbb{Z}_+$ such that $\vex+\gamma_{\veg}\veg$ is still feasible and $\vecc^T(\vex+\gamma_{\veg}\veg)$ is minimized, and among all the $\gamma_{\veg}\veg$ we pick the best one. For any fixed $\veg$ we can easily find $\gamma_{\veg}$ by solving an integer programming with only one integral variable $\gamma_{\veg}$. Therefore the overall running time depends on the cardinality of the Graver bais $\mathcal{G}(A)$. Unfortunately $|\mathcal{G}(A)|$ could be huge in general. However, if the matrix $A$ has some special structure, then $|\mathcal{G}(A)|$ could be significantly smaller. 

From now on we focus on a tree-fold matrix $A$ consisting of $n$ copies of submatrices $A_1$, $A_2$, $\cdots$, $A_{\tau}$ and write it as $A=T[A_1,A_2,\cdots,A_{\tau}]$ for simplicity. Recall that each $A_i$ is an $s_i\times t$-matrix, whereas we are restricting to the following
\begin{eqnarray*}
\min\{\vecc^T\vex| A\vex=\veb, \vel\le\vex\le\veu, \vex\in\mathbb{Z}^{nt}\}.
\end{eqnarray*} 
Notice that if $\tau=2$, $A$ is called an $n$-fold matrix. 
In 2013, Hemmecke et al. provided a Graver-best augmentation algorithm for $n$-fold integer programming that runs in $O(n^3L)$ time (here the big-$O$ hides all coefficients that only depend on $A_1$ and $A_2$). The following lemma is the key ingredient to their algorithm. It strengthens the {\em fitness theorem} in~\cite{hocsten2007finiteness}.

Consider any $\vex\in \mathbb{Z}^{nt}$. We write $\vex$ as a tuple $\vex=(\vex^1,\vex^2,\cdots,\vex^n)$ where $\vex^i\in\mathbb{Z}^t$. Each $\vex^i$ is called a {\em brick} of $\vex$.
\begin{lemma}[\cite{hemmecke2013n}]
	Let $A=T[A_1,A_2]$. There exists some integer $\lambda=\lambda(A_1,A_2)$ that only depends on matrices $A_1$ and $A_2$, and
	$$H(A)=\{\veh\in \mathbb{Z}^t| \veh \text{ is the sum of at most $\lambda$ elements of } \mathcal{G}(A_2)\},$$
	such that for any $\veg=(\veg^1,\veg^2,\cdots,\veg^n)\in\mathcal{G}(A)$
	we have $\sum_{i\in I}\veg^i\in H(A)$ for any $I\subseteq \{1,2,\cdots,n\}$.
\end{lemma}

We further generalize the algorithm of Hemmecke et al.~\cite{hemmecke2013n} to tree-fold integer programming. Towards this, we first give a generalization of the above lemma, and then we show how to further generalize their algorithm.

\subsection{Proof of Lemma~\ref{le:fitness}}\label{ap-sec:fitness}
\begin{proof}[Proof of Lemma~\ref{le:fitness}]
	Throughout this proof, for an arbitrary matrix $B$, we list its Graver bases (in an arbitrary order) as $\veg^1(B),\veg^2(B),\cdots,\veg^{|\G(B)|}(B)$, and let $\veG(B)=(\veg^1(B),\veg^2(B),\cdots,\veg^{|\G(B)|}(B))$ be the matrix with each of the bases being its column.
	
	Consider $A_{\tau}$. For any $\veg=(\veg^1,\veg^2,\cdots,\veg^n)\in\mathcal{G}(A)$, it follows directly that $A_{\tau}\veg^i=0$ for every $1\le i\le n$. According to the positive sum property of the Graver basis, there exist $q_j^i(A_{\tau})\in \mathbb{Z}_{\ge 0}$ such that 
	\begin{eqnarray}\label{eq:veg}
	\veg^i=\sum_{j=1}^{|\G(A_{\tau})|}q^i_j(A_{\tau})\veg^j(A_{\tau})=\veG(A_{\tau})\veq^i(A_{\tau}),\quad \forall 1\le i\le d_{\tau}=n
	\end{eqnarray}
	where $\veq^i(A_{\tau})=(q^i_1(A_{\tau}),\cdots,q^i_{|\G(A_{\tau})|})^T$.
	Notice that $\veg(A_{\tau})$ is only dependent on matrix $A_{\tau}$. In order to show that $\sum_{i\in I}\veg^i\in H(A)$ for some $\lambda$, it suffices to show that $\sum_i||\veq^i(A_{\tau})||_1=\sum_{i,j}|q^i_j(A_{\tau})|$ is upper bounded by some value that only depends on $A_1,A_2,\cdots,A_{\tau}$.
	
	\noindent\textbf{Step 1.} We consider $A_{\tau-1}$. According to $A\veg=0$, we have 
	$$\sum_{i\in S_{\tau-1}^\ell} A_{\tau-1}\veg^i=0, \quad \forall 1\le \ell\le d_{\tau-1}$$
	Plugging in Equation~\ref{eq:veg}, we have 
	\begin{eqnarray}\label{eq:lemma2-eq2}
	\sum_{i\in S_{\tau-1}^\ell} A_{\tau-1}\veG(A_{\tau})\veq^i(A_{\tau})=0. \quad \forall 1\le \ell\le d_{\tau-1}
	\end{eqnarray}
	We rewrite the above equation in the following way. Let matrix $A_{\tau-1}'=A_{\tau-1}\veG(A_{\tau})$, $\veQ^{\ell}(A_{\tau})=\sum_{i\in S_{\tau-1}^\ell}\veq^i(A_{\tau})$, we have 
	\begin{eqnarray}\label{eq:lemma2-eq3}
	\sum_{i\in S_{\tau-1}^\ell}\veg^i=\veG(A_{\tau})\veQ^\ell(A_{\tau}), \quad\forall 1\le \ell\le d_{\tau-1}
	\end{eqnarray}
	\begin{eqnarray}\label{eq:lemma2-eq4}
	A_{\tau-1}'\veQ^{\ell}(A_{\tau})=0, \quad\forall 1\le \ell\le d_{\tau-1}
	\end{eqnarray}
	Therefore, $\veQ^\ell(A_{\tau})\in ker_{\mathbb{Z}^{|\G(A_{\tau})|}}({A_{\tau-1}'})$. We replace the index $\ell$ by $i$. According to the positive sum property, we list the Graver basis of $A_{\tau-1}'$ as $\veg^1(A_{\tau-1}')$, $\cdots$, $\veg^{|\G(A_{\tau-1}')|}(A_{\tau-1}')$, then there exist $q_j^{i}(A_{\tau-1}')\in \mathbb{Z}_{\ge 0}$ such that
	\begin{eqnarray}\label{eq:lemma2-eq5}
	\veQ^{i}(A_{\tau})=\sum_{j=1}^{|\G(A_{\tau-1}')|} q_j^i(A_{\tau-1}')\veg^j(A_{\tau-1}')=\veG(A_{\tau-1}')\veq^i(A_{\tau-1}'), \quad\forall 1\le i\le d_{\tau-1}
	\end{eqnarray}
	where $\veq^i(A_{\tau-1}')=(q^i_1(A_{\tau-1}'),\cdots,q^i_{|\G(A_{\tau-1}')|})^T$. Furthermore, as every entry of $\veQ^i(A_{\tau})$ is non-negative, the positive sum property ensures that $q_j^{i}(A_{\tau-1}')>0$ only if every entry of $\veg^j(A_{\tau-1}')$ is non-negative.
	
	\noindent\textbf{Step 2.} We consider $A_{\tau-2}$. According to $A\veg=0$, we have
	$$\sum_{i_1\in S_{\tau-2}^\ell}\sum_{i_0\in S_{\tau-1}^{i_1}} A_{\tau-2}\veg^{i_0}=0, \quad \forall 1\le \ell\le d_{\tau-2}.$$
	
	Plugging in Equation~\ref{eq:lemma2-eq3} and~\ref{eq:lemma2-eq5}, we have
	$$\sum_{i_1\in S_{\tau-2}^\ell}\sum_{i_0\in S_{\tau-1}^{i_1}} A_{\tau-2}\veg^{i_0}=\sum_{i_1\in S_{\tau-2}^\ell}A_{\tau-2}\veG(A_{\tau})\veQ^{i_1}(A_{\tau})=\sum_{i_1\in S_{\tau-2}^\ell}A_{\tau-2}\veG(A_{\tau})\veG(A_{\tau-1}')\veq^{i_1}(A_{\tau-1}')=0,\quad\forall 1\le \ell\le d_{\tau-2}$$
	
	Let $A_{\tau-2}'=A_{\tau-2}\veG(A_{\tau})\veG(A_{\tau-1}')$, $\veQ^{\ell}(A_{\tau-1}')=\sum_{i_1\in S_{\tau-2}^\ell}\veq^{i_1}(A_{\tau-1}')$, we have
	\begin{eqnarray}\label{eq:lemma2-eq6}
	\sum_{i_1\in S_{\tau-2}^\ell}\sum_{i_0\in S_{\tau-1}^{i_1}}\veg^{i_0}=\veG(A_{\tau})\veG(A_{\tau-1}')\veQ^\ell(A_{\tau-1}'), \quad\forall 1\le \ell\le d_{\tau-2}
	\end{eqnarray}
	\begin{eqnarray}\label{eq:lemma2-eq7}
	A_{\tau-2}'\veQ^{\ell}(A_{\tau-1}')=0, \quad\forall 1\le \ell\le d_{\tau-2}
	\end{eqnarray}
	Therefore, $\veQ^{\ell}(A_{\tau-1}')\in ker_{\mathbb{Z}^{|\G(A_{\tau-1})'|}}(A_{\tau-2}')$. Replacing the index $\ell$ by $i$, there exist $q_j^i(A_{\tau-2}')\in\mathbb{Z}_{\ge0}$ such that 
	\begin{eqnarray}\label{eq:lemma2-eq8}
	\veQ^{i}(A_{\tau-1}')=\sum_{j=1}^{|\G(A_{\tau-2}')|} q_j^i(A_{\tau-2}')\veg^j(A_{\tau-2}')=\veG(A_{\tau-2}')\veq^i(A_{\tau-2}'), \quad\forall 1\le i\le d_{\tau-2}
	\end{eqnarray}
	where $\veq^i(A_{\tau-2}')=(q^i_1(A_{\tau-2}'),\cdots,q^i_{|\G(A_{\tau-2}')|})^T$.
	
	We can iteratively carry on the above argument. 
	
	\noindent\textbf{Step $\tau-k$.} In general, suppose we have shown the following three equations:
	\begin{eqnarray}\label{eq:lemma2-eq9}
	\sum_{i_{\tau-k-2}\in S_{k+1}^\ell}\sum_{i_{\tau-k-3}\in S_{k+2}^{i_{\tau-k-2}}}\cdots\sum_{i_{0}\in S_{\tau-1}^{i_{1}}}\veg^{i_{0}}=\veG(A_{\tau})\veG(A_{\tau-1}')\cdots\veG(A_{k+2}')\veQ^\ell(A_{k+2}'), \quad\forall 1\le \ell\le d_{k+1}
	\end{eqnarray}
	\begin{eqnarray}\label{eq:lemma2-eq10}
	A_{k+1}'\veQ^{\ell}(A_{k+2}')=0, \quad\forall 1\le \ell\le d_{k+1}
	\end{eqnarray}
	Replacing the index $\ell$ by $i$, there exist $q_j^i(A_{k+1}')\in \mathbb{Z}_{\ge 0}$ such that
	\begin{eqnarray}\label{eq:lemma2-eq11}
	\sum_{i'\in S_{k+1}^i}\veq^{i'}(A_{k+2}')=\veQ^{i}(A_{k+2}')=\sum_{j=1}^{|\G(A_{k+1}')|} q_j^i(A_{k+1}')\veg^j(A_{k+1}')=\veG(A_{k+1}')\veq^i(A_{k+1}'), \quad\forall 1\le i\le d_{k+1}
	\end{eqnarray}
	where $\veq^i(A_{k+1}')=(q^i_1(A_{k+1}'),\cdots,q^i_{|\G(A_{k+1}')|})^T$,
	and $A_{k+1}'=A_{k+1}\veG(A_{\tau})\veG(A_{\tau-1}')\cdots\veG(A_{k+2}')$.
	
	When we consider $A_k$, $A\veg=0$ implies that
	\begin{eqnarray}\label{eq:lemma2-eq12}
	\sum_{i_{\tau-k-1}\in S_{k}^\ell}\sum_{i_{\tau-k-2}\in S_{k+1}^{\tau-k-1}}\sum_{i_{\tau-k-3}\in S_{k+2}^{i_{\tau-k-2}}}\cdots\sum_{i_{0}\in S_{\tau-1}^{i_{1}}} A_{k}\veg^{i_0}=0, \quad \forall 1\le \ell\le d_{k}.
	\end{eqnarray}
	Indeed, if we view each $\veg^i$ as the $i$-th leaf (from left to right) of the tree, the summation is taken over all the leaves of the sub-tree routed at the vertex corresponding to $S_k^{\ell}$.
	Plugging Equation~\ref{eq:lemma2-eq9} and Equation~\ref{eq:lemma2-eq11} into Equation~\ref{eq:lemma2-eq12}, and replacing index $i_{\tau-k-1}$ by $i$, we have
	\begin{eqnarray}\label{eq:lemma2-eq13}
	\sum_{i\in S_{k}^\ell}A_{k} \veG(A_{\tau})\veG(A_{\tau-1}')\cdots\veG(A_{k+2}')\veG(A_{k+1}')\veq^{i}(A_{k+1}')=0, \quad \forall 1\le \ell\le d_{k}
	\end{eqnarray}
	Let $A_{k}'=A_{k} \veG(A_{\tau})\veG(A_{\tau-1}')\cdots\veG(A_{k+1}')$, $\veQ^{\ell}(A_{k+1}')=\sum_{i'\in S_{k}^\ell}\veq^{i'}(A_{k+1}')$, we have
	\begin{eqnarray}\label{eq:lemma2-eq14}
	\sum_{i_{\tau-k-1}\in S_{k}^\ell}\sum_{i_{\tau-k-2}\in S_{k+1}^{\tau-k-1}}\cdots\sum_{i_{0}\in S_{\tau-1}^{i_{1}}} \veg^{i_0}=\veG(A_{\tau})\veG(A_{\tau-1}')\cdots\veG(A_{k+1}')\veQ^\ell(A_{k+1}'), \quad\forall 1\le \ell\le d_{k}
	\end{eqnarray}
	\begin{eqnarray}\label{eq:lemma2-eq15}
	A_{k}'\veQ^{\ell}(A_{k+1}')=0, \quad\forall 1\le \ell\le d_{k}
	\end{eqnarray}
	Therefore, $\veQ^{\ell}(A_{k+1}')\in ker_{\mathbb{Z}^{|\G(A_{k+1})'|}}(A_{k}')$. Replacing the index $\ell$ by $i$, there exist $q_j^i(A_{k}')\in\mathbb{Z}_{\ge0}$ (by the positive sum property) such that 
	\begin{eqnarray}\label{eq:lemma2-eq16}
	\sum_{i'\in S_{k}^i}\veq^{i'}(A_{k+1}')=\veQ^{i}(A_{k+1}')=\sum_{j=1}^{|\G(A_{k}')|} q_j^i(A_{k}')\veg^j(A_{k}')=\veG(A_{k}')\veq^i(A_{k}'), \quad\forall 1\le i\le d_{k}
	\end{eqnarray}
	where $\veq^i(A_{k}')=(q^i_1(A_{k}'),\cdots,q^i_{|\G(A_{k}')|})^T$,
	and $A_{k}'=A_{k}\veG(A_{\tau})\veG(A_{\tau-1}')\cdots\veG(A_{k+1}')$. 
	
	Specifically, we let $A_{\tau}'=A_{\tau}$, therefore the above equalities hold for any $1\le k\le \tau-1$.
	
	\noindent\textbf{Step $\tau-1$.}
	Eventually we consider $A_1$ and derive the following based on the iterative argument.
	\begin{eqnarray}\label{eq:lemma2-eq17}
	\sum_{i_{\tau-2}\in S_{1}^\ell}\sum_{i_{\tau-3}\in S_{2}^{i_{\tau-2}}}\cdots\sum_{i_{0}\in S_{\tau-1}^{i_{1}}}\veg^{i_{0}}=\veG(A_{\tau})\veG(A_{\tau-1}')\cdots\veG(A_{2}')\veQ^\ell(A_{2}'), \quad\forall 1\le \ell\le d_{1}=1
	\end{eqnarray}
	\begin{eqnarray}\label{eq:lemma2-eq18}
	A_{1}'\veQ^{\ell}(A_{2}')=0, \quad\forall 1\le \ell\le d_{1}=1
	\end{eqnarray}
	Replacing the index $\ell$ by $i$, there exist $q_j^i(A_{1}')\in\mathbb{Z}_{\ge 0}$ such that
	\begin{eqnarray}\label{eq:lemma2-eq19}
	\sum_{i'\in S_{1}^i}\veq^{i'}(A_{2}')=\veQ^{i}(A_{2}')=\sum_{j=1}^{|\G(A_{1}')|} q_j^i(A_{1}')\veg^j(A_{1}')=\veG(A_{1}')\veq^i(A_{1}'), \quad\forall 1\le i\le d_{1}=1
	\end{eqnarray}
	where $\veq^i(A_{1}')=(q^i_1(A_{1}'),\cdots,q^i_{|\G(A_{1}')|})^T$,
	and $A_{1}'=A_{1}\veG(A_{\tau})\veG(A_{\tau-1}')\cdots\veG(A_{2}')$.
	
	We make the following claim.
	\begin{claim}
		$\veQ^{i}(A_2')\in \G(A_1')$.
	\end{claim}
	\begin{proof}[Proof of the Claim] Suppose on the contrary that $\veQ^{i}(A_2')\not\in \G(A_1')$, then there exist $0\neq \bar{\veQ}^{i}(A_2')\sqsubset \veQ^{i}(A_2')$ such that $A_{1}'\bar{\veQ}^{i}(A_{2}')=0$. In the following we will construct $0\neq \bar{\veg}\sqsubset \veg$ such that $A\bar{\veg}=0$, which contradicts the fact that $\veg\in \G(A)$. Hence, the claim is true.
		
		We show how to construct $\bar{\veg}$. According to Equation~\ref{eq:lemma2-eq19}, $\sum_{i'\in S_{1}^i}\veq^{i'}(A_{2}')=\veQ^{i}(A_{2}')$. We know that every entry of $\veq^{i'}(A_{2}')$, and consequently $\veQ^{i}(A_{2}')$, is non-negative. Therefore every entry of $\bar{\veQ}^{i}(A_2')$ is also non-negative. Consider every entry of the equation $\sum_{i'\in S_{1}^i}\veq^{i'}(A_{2}')=\veQ^{i}(A_{2}')$, we have $\sum_{i'\in S_{1}^i}\veq^{i'}_j(A_{2}')=\veQ^{i}_j(A_{2}')$. For $0\le \bar{\veQ}^{i}_j(A_{2}')\le \veQ^{i}_j(A_{2}')$, we can easily find $0\le \bar{\veq}^{i'}_j(A_{2}')\le \veq^{i'}_j(A_{2}')$ such that $\sum_{i'\in S_{1}^i}\bar{\veq}^{i'}_j(A_{2}')=\bar{\veQ}^{i}_j(A_{2}')$. Hence, there exist $\bar{\veq}^{i'}(A_{2}')\sqsubseteq \veq^{i'}(A_{2}')$ such that 
		$$\sum_{i'\in S_{1}^i}\bar{\veq}^{i'}(A_{2}')=\bar{\veQ}^{i}_j(A_{2}'),\quad \forall 1\le i\le d_1=1,$$
		and moreover, there exist some $i_1'$ and $i_2'$ such that $\bar{\veq}^{i'_1}(A_{2}')\sqsubset \veq^{i'_1}(A_{2}')$ and $\bar{\veq}^{i'_2}(A_{2}')\neq 0$.
		
		Replacing $i'$ with $i$, we define 
		$$\bar{\veQ}^{i}(A_{3}')=\veG(A_{2}')\bar{\veq}^i(A_{2}'),\quad 1\le i\le d_2.$$
		It is easy to see that $\bar{\veQ}^{i}(A_{3}')\sqsubseteq {\veQ}^{i}(A_{3}')$ for $1\le i\le d_2$. As each $\bar{\veQ}^{i}(A_{3}')$ is the weighted sum of the Graver basis of $A_2'$, we know $A_2'\bar{\veQ}^{i}(A_{3}')=0$. Furthermore, there exist $1\le i_1,i_2\le d_2$ such that $\bar{\veQ}^{i_1}(A_{3}')\sqsubset {\veQ}^{i_1}(A_{3}')$ and $\bar{\veQ}^{i_2}(A_{3}')\neq 0$.
		
		Carry on the above argument, we can prove iteratively that there exist $\bar{\veq}^{i'}(A_{k+1}')\sqsubseteq \veq^{i'}(A_{k+1}')$ such that 
		$$\sum_{i'\in S_{k}^i}\bar{\veq}^{i'}(A_{k+1}')=\bar{\veQ}^{i}_j(A_{k+1}'),\quad \forall 1\le i\le d_k.$$
		Furthermore, there exist some $i_1'$ and $i_2'$ such that $\bar{\veq}^{i'_1}(A_{k+1}')\sqsubset \veq^{i'_1}(A_{2}')$ and $\bar{\veq}^{i'_2}(A_{k+1}')\neq 0$. 
		
		Replacing the index $i'$ with $i$, we define 
		$$\bar{\veQ}^{i}(A_{k+2}')=\veG(A_{k+1}')\bar{\veq}^i(A_{k+1}'),\quad 1\le i\le d_{k+1}.$$ 
		Then $\bar{\veQ}^{i}(A_{k+2}')\sqsubseteq {\veQ}^{i}(A_{k+2}')$ for $1\le i\le d_{k+1}$. As each $\bar{\veQ}^{i}(A_{k+2}')$ is the weighted sum of the Graver basis of $A_{k+1}'$, we know $A_{k+1}'\bar{\veQ}^{i}(A_{k+2}')=0$. Furthermore, there exist $1\le i_1,i_2\le d_{k+1}$ such that $\bar{\veQ}^{i_1}(A_{k+2}')\sqsubset {\veQ}^{i_1}(A_{k+2}')$ and $\bar{\veQ}^{i_2}(A_{k+2}')\neq 0$.
		
		Eventually, we can show that there exist $\bar{\veQ}^{i}(A_{\tau})=\veG(A_{\tau-1}')\bar{\veq}^i(A_{\tau-1}')$ for $1\le i\le d_{\tau-1}$ such that $\bar{\veQ}^{i}(A_{\tau})\sqsubseteq {\veQ}^{i}(A_{\tau})$, $A_{\tau-1}'\bar{\veQ}^{i}(A_{\tau})=0$. Furthermore, there exist $1\le i_1,i_2\le d_{\tau-1}$ such that $\bar{\veQ}^{i_1}(A_{\tau})\sqsubset {\veQ}^{i_1}(A_{\tau})$ and $\bar{\veQ}^{i_2}(A_{\tau})\neq 0$.
		
		Given that $\veQ^i(A_{\tau})=\sum_{i'\in S_{\tau-1}^i} \veq^{i'}(A_{\tau})$, we can find $\bar{\veq}^{i'}(A_{\tau})\sqsubseteq \veq^{i'}(A_{\tau})$ such that $\bar{\veQ}^i(A_{\tau})=\sum_{i'\in S_{\tau-1}^i} \bar{\veq}^{i'}(A_{\tau})$, and moreover, there exist $1\le i_1,i_2\le n$ such that $\bar{\veq}^{i_1}(A_{\tau})\sqsubset {\veq}^{i_1}(A_{\tau})$ and $\bar{\veq}^{i_2}(A_{\tau})\neq 0$. 
		
		We define
		\begin{eqnarray*}
			\bar{\veg}^i=\veG(A_{\tau})\bar{\veq}^i(A_{\tau}),\quad \forall 1\le i\le d_{\tau}=n
		\end{eqnarray*}
		Note that by the positive sum property of the Graver basis, if $q_j^i(A_{\tau})> 0$ then $\veg^j(A_{\tau})$ must lie in the same orthant as $\veg^i$. Therefore $\bar{\veq}^{i'}(A_{\tau})\sqsubseteq \veq^{i'}(A_{\tau})$ implies that $\bar{\veg}^i\sqsubseteq \veg^i$. Further, $A_{\tau}\bar{\veg}^i=0$, and moreover, there exist $1\le i_1,i_2\le n$ such that $\bar{\veg}^{i_1}\sqsubset {\veg}^{i_1}$ and $\bar{\veg}^{i_2}\neq 0$. Therefore, $0\neq \bar{\veg}=(\bar{\veg}^1,\cdots,\bar{\veg}^n)\sqsubset \veg$. 
		
		Finally we show that $A\bar{\veg}=0$. This is equivalent as showing for every $1\le k\le \tau-1$, 
		\begin{eqnarray*}
			\sum_{i_{\tau-k-1}\in S_{k}^\ell}\sum_{i_{\tau-k-2}\in S_{k+1}^{\tau-k-1}}\sum_{i_{\tau-k-3}\in S_{k+2}^{i_{\tau-k-2}}}\cdots\sum_{i_{0}\in S_{\tau-1}^{i_{1}}} A_{k}\bar{\veg}^{i_0}=0, \quad \forall 1\le \ell\le d_{k}.
		\end{eqnarray*}
		
		Using the equations $\bar{\veg}^i=\veG(A_{\tau})\bar{\veq}^i(A_{\tau})$ and $\sum_{i'\in S_{k}^i}\bar{\veq}^{i'}(A_{k+1}')=\bar{\veQ}^{i}(A_{k+1}')=\veg(A_k')\bar{\veq}^i(A_{k}')$, we have
		\begin{eqnarray*}
			&&\sum_{i_{\tau-k-1}\in S_{k}^\ell}\sum_{i_{\tau-k-2}\in S_{k+1}^{\tau-k-1}}\sum_{i_{\tau-k-3}\in S_{k+2}^{i_{\tau-k-2}}}\cdots\sum_{i_{0}\in S_{\tau-1}^{i_{1}}} A_{k}\bar{\veg}^{i_0}\\
			&=&\sum_{i_{\tau-k-1}\in S_{k}^\ell}\sum_{i_{\tau-k-2}\in S_{k+1}^{\tau-k-1}}\sum_{i_{\tau-k-3}\in S_{k+2}^{i_{\tau-k-2}}}\cdots\sum_{i_{0}\in S_{\tau-1}^{i_{1}}} A_{k}\veG(A_{\tau})\bar{\veq}^{i_0}(A_{\tau})\\
			&=&\sum_{i_{\tau-k-1}\in S_{k}^\ell}\sum_{i_{\tau-k-2}\in S_{k+1}^{\tau-k-1}}\sum_{i_{\tau-k-3}\in S_{k+2}^{i_{\tau-k-2}}}\cdots\sum_{i_{1}\in S_{\tau-2}^{i_{2}}} A_{k}\veG(A_{\tau})\bar{\veQ}^{i_1}(A_{\tau})\\
			&=&\sum_{i_{\tau-k-1}\in S_{k}^\ell}\sum_{i_{\tau-k-2}\in S_{k+1}^{\tau-k-1}}\sum_{i_{\tau-k-3}\in S_{k+2}^{i_{\tau-k-2}}}\cdots\sum_{i_{1}\in S_{\tau-2}^{i_{2}}} A_{k}\veG(A_{\tau})\veG(A_{\tau-1}')\bar{\veq}^{i_1}(A_{\tau})\\
			&=&\sum_{i_{\tau-k-1}\in S_{k}^\ell}\sum_{i_{\tau-k-2}\in S_{k+1}^{\tau-k-1}}\sum_{i_{\tau-k-3}\in S_{k+2}^{i_{\tau-k-2}}}\cdots\sum_{i_{1}\in S_{\tau-3}^{i_{3}}} A_{k}\veG(A_{\tau})\veG(A_{\tau-1}')\bar{\veQ}^{i_2}(A_{\tau-1}')\\
			&=&...\\
			&=&A_{k}\veG(A_{\tau})\veG(A_{\tau-1}')\cdots\veG(A_{k+1}')\bar{\veQ}^\ell=A_k'\bar{\veQ}^\ell=0
		\end{eqnarray*}
		Therefore, the claim is proved.
	\end{proof}
	We now show that $\sum_{i=1}^n||\veq^i(A_{\tau})||_1=\sum_{i,j}|q^i_j(A_{\tau})|$ is upper bounded by some value that only depends on $A_1,A_2,\cdots,A_{\tau}$. Using the fact that $\sum_{i'\in S_{k}^i}{\veq}^{i'}(A_{k+1}')={\veQ}^{i}(A_{k+1}')=\veG(A_k'){\veq}^i(A_{k}')$,we have
	\begin{eqnarray*}
		\sum_{i=1}^{d_{k+1}}||\veq^i(A_{k+1}')||_1
		=\sum_{i=1}^{d_{k}}||\veQ^i(A_{k+1}')||_1
		=\sum_{i=1}^{d_k}||\veG(A_{k}')\veq^i(A_{k}')||_1 \le\sum_{i=1}^{d_k}||\veG(A_{k}')||_1||\veq^i(A_{k}')||_1
		=\sum_{i=1}^{d_{k-1}}||\veG(A_{k}')||_1||\veQ^i(A_{k}')||_1
	\end{eqnarray*}
	
	Therefore, $$\sum_{i=1}^n||\veq^i(A_{\tau})||_1\le ||\veG(A_{\tau-1}')||_1||\veG(A_{\tau-2}')||_1\cdots||\veG(A_{2}')||_1||\veQ^i(A_{2}')||_1.$$
	Obviously each $A_k'$, and hence its Graver basis, and hence $||\veG(A_k')||$, is only dependent on $A_1,\cdots, A_{\tau}$. Furthermore, $\veQ^i(A_{2}')\in \G(A_1')$, hence $||\veQ^i(A_{2}')||_1$, and consequently $\sum_{i=1}^n||\veq^i(A_{\tau})||_1$, is only dependent on $A_1,\cdots, A_{\tau}$. Thus, for $\lambda=||\veG(A_{\tau-1}')||_1||\veG(A_{\tau-2}')||_1\cdots||\veG(A_{2}')||_1||\veQ^i(A_{2}')||_1$ we have $\veg\in H(A)$, and the lemma is proved.
\end{proof}

\subsection{Proof of Lemma~\ref{lemma:cite}}\label{ap-sec:cite}
\begin{proof}[Proof of Lemma~\ref{lemma:cite}]
	Notice that if we fix $\veg=\veg^*$, then $\gamma=\gamma^*$ is the largest integer such that $\vel\le \vex+\gamma\veg^*\le\veu$ is still true. Therefore, if we consider each brick of the solution $\vex=(\vex^1,\vex^2,\cdots,\vex^n)$, then there exists some $1\le i\le n$ such that $\gamma^*$ is the largest integer such that $\vel^i\le \vex^i+\gamma\veg^{*i}\le\veu^i$ is still true. As $\veg^{*}\in H(A)$, $\veg^{*i}\in H(A)$ for every $i$. Now for every $\veh\in H(A)$ and every $1\le i\le n$, we find out the largest integer $\gamma_{\veh,i}$ such that $\vel^i\le \vex^i+\gamma_{\veh,i}\veh^i\le\veu^i$ is true and add this integer to $\Gamma$. Obviously $\gamma^*\in \Gamma$ and $|\Gamma|\le n|H(A)|$.
	\end{proof}

\subsection{Constructing an initial feasible solution}
\label{ap-sec:initial-solution}
We have proved the correctness of Theorem~\ref{thm:main-ILP} if a feasible initial solution is given. In case a feasible solution is unknown,
we construct an auxiliary tree-fold integer programming such that i). the initial feasible solution of the auxiliary programming is trivial; ii). the optimal solution of the auxiliary programming gives a feasible initial solution for the original tree-fold programming~(\ref{eq:tree-fold}). The argument is essentially the same as that of~\cite{hemmecke2013n}.

We add auxiliary variables. For each $\vex^i$, we add $2\sum_{k=1}^{\tau}s_k$ auxiliary variables and let them be $\vez^i$. The new vector of variables becomes $(\vex^1,\vez^1,\vex^2,\vez^2,\cdots,\vex^n,\vez^n)$.

We introduce a lower bound of $0$ and upper bound of $||\veb||_{\infty}$ for each auxiliary variable. For each $1\le k\le \tau$, we replace each $A_k$ with $(A_k,0_{s_k\times s_1},0_{s_k\times s_1},0_{s_k\times s_2},0_{s_k\times s_2},0_{s_k\times s_3},\cdots,0_{s_k\times s_{k-1}},I_{s_k\times s_k},-I_{s_k\times s_k},0_{s_k\times s_{k+1}},0_{s_k\times s_{k+1}},\cdots,0_{s_k\times s_{\tau}})$.

We change the objective function as the summation of all the auxiliary variables.

A feasible initial solution for the auxiliary ILP could be easily derived by setting $\vex=0$ and approperiate values to the auxiliary variables. Furthermore, the optimal solution of the auxiliary ILP is $0$ if and only if there exists a feasible solution for~(\ref{eq:tree-fold}). Therefore, we can apply our algorithm of the previous subsection to solve the auxiliary ILP and derive its optimal solution, which provides an initial feasible solution for the original tree-fold integer programming~(\ref{eq:tree-fold}).

\subsection{Proof of Lemma~\ref{lemma:ILP-to-solution}}\label{ap-sec:ILP-to-solution}
\begin{proof}[Proof of Lemma~\ref{lemma:ILP-to-solution}]
	Let $LF$ be the set of leaves. In the following we show that it is possible to select a subset $LF'\subseteq LF$ such that there exists a subtree of weight at most $B$ that contains each vertex of $LF'$, and furthermore, if we delete $LF'$ (together with the edge incident to them) from the tree $T$, there exists a feasible solution of the ILP for the remaining tree $T'$ with the objective value at most $m-1$. If the above claim is true, we can iteratively carry on the argument to construct $m$ subtrees that contain every vertex of $LF$ and the lemma is proved.
	
	We pick an arbitrary $j_0$ such that  $x_{1,(CF_{j_0},1)}\ge 1$. Consider the children of the root $v_1$. According to constraint $(I)$, for any location $k$ such that $f_{j_0}(k)=1$ (i.e., the location of the vertices who are children of the root of $CF_{j_0}$), we have $$\sum_{s:v_s\in CH(v_1)} x_{s,(CF_{j_0},k)}=x_{1,(CF_{j_0},1)}\ge 1.$$
	Hence, for any $k$ such that $f_{j_0}(k)=1$, there exists at least one child of $v_1$, say, $v_{s(1,k)}$, such that $x_{s(1,k),(CF_{j_0},k)}\ge 1$. We pick an arbitrary one (if there are multiple) of such vertices for every $k$ and let $H(1)$ be the set of these vertices.
	
	Consider an arbitrary $v_{s(k_1)}\in H_1$ where $x_{s(k_1),(CF_j,k_1)}\ge 1$. According to constraint $(I)$, for any $k_2$ such that $f_{j_0}(k_2)=k_1$, we have 
	$$\sum_{s:v_s\in CH(v_{s(k_1)})} x_{s,(CF_{j_0},k_2)}=x_{{s(k_1)},(CF_{j_0},k_1)}\ge 1.$$
	Hence, for any $k_2$ such that $f_{j_0}(k_2)=k_1$, there exists at least one child of $v_{s(k_1)}$, say, $v_{s(k_2)}$ such that $x_{s(k_2),(CF_{j_0},k_2)}\ge 1$. We pick an arbitrary one of such vertices for every $k_2$ such that $f_{j_0}(k_2)=k_1$, and let $H(1,k_1)$ be the set of these vertices.
	
	Suppose in general we have constructed the set of vertices $H(1,k_1,k_2,\cdots,k_i)$ such that
	\begin{itemize}
		\item for any $1\le h\le i$, $f_{j_0}(k_{h})=k_{h-1}$;
		\item for any $k_{i+1}$ such that $f_{j_0}(k_{i+1})=k_i$, there exists exactly one vertex $v_{s(k_{i+1})}\in H(1,k_1,k_2,\cdots,k_i)$ such that $x_{s(k_{i+1}),(CF_{j_0},k_{i+1})}\ge 1$.
	\end{itemize} 
	If there exists at least one vertex of $H(1,k_1,k_2,\cdots,k_i)$ which is not a leaf, we proceed as follows. For any $v_{s(k_{i+1})}\in H(1,k_1,\cdots,k_i)$ which is not a leaf and any $k_{i+2}$ such that $f(k_{i+2})=k_{i+1}$, the following is true:
	$$\sum_{s:v_s\in CH(v_{s(k_{i+1})})} x_{s,(CF_{j_0},k_{i+2})}=x_{{s(k_{i+1})},(CF_{j_0},k_{i+1})}\ge 1.$$
	Hence, there exists at least one child of $v_{s(k_{i+1})}$, say, $v_{s(k_{i+2})}$ such that $x_{s(k_{i+2}),(CF_{j_0},k_{i+2})}\ge 1$. We pick an arbitrary one of such vertices for every $k_{i+2}$ and let $H(1,k_1,\cdots,k_{i+1})$ be the set of them. Otherwise every vertex of $H(1,k_1,k_2,\cdots,k_i)$ is a leaf and we stop.
	
	Eventually we derive a sequence of sets $H(1,k_1,k_2,\cdots,k_i)$ and let $H$ be the union of them.
	
	Let $T[H]$ be the induced subgraph of $T$. Firstly, we claim that $T[H]$ is a subtree of the original tree $T$. To see why, it suffices to notice that every vertex of $H(1,k_1,k_2,\cdots,k_i)$ is connected to the root $v_1$.
	
	Secondly, we claim that every leaf of the subtree $T[H]$ is also a leaf in $T$. This is straightforward. Let $v_s$ be an arbitrary leaf of $T[H]$ which is not a leaf in the original graph, then according to our iterative construction, we will further consider the children of $v_s$ and add some of them to $H$. 
	
	Thirdly, we claim that the weight of $T[H]$ is at most $B$. Indeed, the claim follows directly as every vertex of $H$ is consistent to some vertex in $CF_{j_0}$. 
	
	Let $LF(H)$ be the set of leaves in $T[H]$. We delete $LF(H)$ and the edges incident to them in $T$ and consider the ILP for the remaining subtree $T'$. It is easy to verify that the following solution $x_{s,(CF_j,k)}'$ is a feasible solution to $ILP(T')$ with the objective of at most $m-1$:
	\begin{eqnarray*}
		&&x_{s,(CF_j,k)}'=x_{s,(CF_j,k)},\quad \textrm{if } j\neq j_0\\
		&&x_{s,(CF_{j_0},k)}'=x_{s,(CF_{j_0},k)}-1,\quad \textrm{if }  v_s\in H\setminus LF(H) 
	\end{eqnarray*}
	
	Therefore given a feasible integer solution with the objective value at most $m$, we can iteratively construct at most $m$ subtrees such that every vertex is covered, and the lemma is proved. 
\end{proof}

\subsection{Tuning the ILP}\label{ap-sec:tuning}
We alter the ILP a bit so that it becomes a tree-fold integer programming.


Given $CF_j$, we let $F^{-1}_j(k)=\{w|f_j(w)=k\}$. For $h\ge 2$, we define $F^{-h}_j(k)=\{w|f_j(w)\in F^{-h+1}_j(k)\}$. Recall that $f_j$ is the function that maps the location of a vertex to the location of its parent in $CF_j$, therefore $F^{-h}_j(k)$ the set of locations of vertices satisfying the following: i). they are descendants of the location $k$ vertex; ii). for each of them, the unweighted distance to the location $k$ vertex is $h$.

We show that, it is possible to remove all the variables $x_{i,(CF_j,k)}$ where $v_i$ is not a leaf and establish an equivalent ILP. 

Let $LF(v_i)$ be the set of all leaves of the subtree rooted at $v_i$. By constraint $(I)$, we have the following
$$x_{i,(CF_j,k)}=\sum_{s:v_s\in CH(v_i)} x_{s,(CF_j,w)},\quad \forall w\in F_j^{-1}(k).$$
If $w\in F_j^{-1}(k)$ is not a leaf, we could further express $x_{s,(CF_j,w)}$ into the summation of other variables. In general, consider any vertex $v_i$ whose depth is $h(T)-h$. As the depth of every leaf is $h(T)$, the unweighted distance of any leaf in $LF(v_i)$ to $v_i$ is $h$, and we have the following:
$$x_{i,(CF_j,k)}=\sum_{s:v_s\in LF(v_i)} x_{s,(CF_j,w)},\quad \forall w\in F_j^{-h}(k).$$
Specifically,
$$x_{1,(CF_j,1)}=\sum_{s:v_s\in LF} x_{s,(CF_j,w)},\quad \forall w\in F_j^{-h(T)}(1).$$

Now every $x_{1,(CF_j,1)}$ could be expressed using $x_{s,(CF_j,w)}$ where $v_s$ is a leaf. We replace the objective function using the above equations.

Let $L_h(CF_j)$ be the subset of locations of $CF_j$ whose depth is $h(T)-h$, and let $L_h^{\ge 2}(CF_j)=\{k||F_j^{-h}(k)|\ge 2\}$, we replace constraint $(I)$ by the following:
$$\sum_{s:v_s\in LF(v_i)} x_{s,(CF_j,w)}-\sum_{s:v_s\in LF(v_i)} x_{s,(CF_j,w')}=0,\quad \forall v_i\in V_h, k\in L_h^{\ge 2}(CF_j), w,w'\in F_j^{-h}(k),\quad (I')$$
where $V_h$ is the set of vertices of depth $h(T)-h$.

It is obvious that the new ILP is equivalent as the original ILP since we simply replace each $x_{s,(CF_j,w)}$ where $v_s$ is not a leaf with the equality it satisfies.


In the following we show that the modified ILP belongs to the tree-fold integer programming. It suffices to consider constraints $(I')$ and $(II)$. Let $\vex^i=$ $(x_{i,(CF_1,1)}$, $x_{i,(CF_1,2)}$, $\cdots$, $x_{i,(CF_1,\zeta)},x_{i,(CF_2,1)},\cdots,x_{i,(CF_2,\zeta)},\cdots,x_{i,(CF_\mu,\zeta)})^T$ and $\vex=(\vex^1,\vex^2,\cdots,\vex^{|LF|})^T$.

Consider constraint $(II)$:
$$\sum_{j=1}^{\mu}\sum_{k=1}^{\zeta} x_{i,(CF_j,k)}=1, \quad \forall v_i\in LF$$
Let $\tau=|h(T)|+1$. We define $A_{1}=I_{\mu\zeta\times\mu\zeta}$, constraint $(II)$ could be written as $\sum_i A_1\vex^i=(1,1,\cdots,1)_{1\times\mu\zeta}$.

Consider constraint $(I')$. For any vertex $v_s\in LF(v_i)$ where $v_i\in V_h$, the constraint $(I')$ could be rewritten as $\sum_{s:v_s\in LF(v_i)}A_{\tau-h}\vex^s=0$ where $A_{\tau-h}$ consists of $\sum_j\sum_{k\in L^{\ge 2}_h(CF_j)}(|F_j^{-h}(k)|-1)\cdot|F_j^h(k)|/2$ different rows, and each row consists of $0,1,-1$ such that the entry that becomes the coefficient of $x_{s,(CF_j,w)}$ after multiplication is $1$, the entry that becomes the coefficient of $x_{s,(CF_j,w')}$ after multiplication is $-1$, and other entries are $0$. Given the fact that $LF(v_i)=\cup_{s:v_s\in CH(v_i)}LF(v_s)$, it is not difficult to verify that contraints $(I')$ and $(II)$ could be written as $A\vex=b$ where $A$ is a tree-fold matrix consisting of submatrices $A_1$, $A_2$, $\cdots$, $A_{\tau}$.

Now applying Theorem~\ref{thm:main-ILP}, an $f(B)n^4$ time algorithm for the subtree cover problem is derived for some function $f$, and Theorem~\ref{thm:main-fpt} is proved.

\clearpage
\bibliographystyle{plain}
\bibliography{schedule-tree}

\begin{thebibliography}{10}

\bibitem{chen2014optimality}
Lin Chen, Klaus Jansen, and Guochuan Zhang.
\newblock On the optimality of approximation schemes for the classical
  scheduling problem.
\newblock In {\em Proceedings of the Twenty-Fifth Annual ACM-SIAM Symposium on
  Discrete Algorithms}, pages 657--668. Society for Industrial and Applied
  Mathematics, 2014.

\bibitem{de2013algebraic}
Jes{\'u}s~A De~Loera, Raymond Hemmecke, and Matthias K{\"o}ppe.
\newblock {\em Algebraic and geometric ideas in the theory of discrete
  optimization}, volume~14.
\newblock SIAM, 2013.

\bibitem{graver1975foundations}
Jack~E Graver.
\newblock On the foundations of linear and integer linear programming i.
\newblock {\em Mathematical Programming}, 9(1):207--226, 1975.

\bibitem{hemmecke2013n}
Raymond Hemmecke, Shmuel Onn, and Lyubov Romanchuk.
\newblock N-fold integer programming in cubic time.
\newblock {\em Mathematical Programming}, 137(1-2):325--341, 2013.

\bibitem{hemmecke2011polynomial}
Raymond Hemmecke, Shmuel Onn, and Robert Weismantel.
\newblock A polynomial oracle-time algorithm for convex integer minimization.
\newblock {\em Mathematical Programming}, 126(1):97--117, 2011.

\bibitem{hochbaum1987using}
Dorit~S Hochbaum and David~B Shmoys.
\newblock Using dual approximation algorithms for scheduling problems
  theoretical and practical results.
\newblock {\em Journal of the ACM (JACM)}, 34(1):144--162, 1987.

\bibitem{hocsten2007finiteness}
Serkan Ho{\c{s}}ten and Seth Sullivant.
\newblock A finiteness theorem for markov bases of hierarchical models.
\newblock {\em Journal of Combinatorial Theory, Series A}, 114(2):311--321,
  2007.

\bibitem{jansen2015structural}
Bart~MP Jansen and Stefan Kratsch.
\newblock A structural approach to kernels for ilps: Treewidth and total
  unimodularity.
\newblock In {\em Algorithms-ESA 2015}, pages 779--791. Springer, 2015.

\bibitem{JansenKV16}
Klaus Jansen, Kim{-}Manuel Klein, and Jos{\'{e}} Verschae.
\newblock Closing the gap for makespan scheduling via sparsification
  techniques.
\newblock In {\em 43rd International Colloquium on Automata, Languages, and
  Programming, {ICALP} 2016, July 11-15, 2016, Rome, Italy}, pages 72:1--72:13,
  2016.

\bibitem{jansen2010technical}
Klaus Jansen and Monaldo Mastrolilli.
\newblock Scheduling unrelated parallel machines: linear programming strikes
  back.
\newblock {\em University of Kiel, Technical Report 1004}, 2010.

\bibitem{knop2016scheduling}
Du{\v{s}}an Knop and Martin Kouteck{\`y}.
\newblock Scheduling meets n-fold integer programming.
\newblock {\em arXiv preprint arXiv:1603.02611}, 2016.

\bibitem{kratsch2016polynomial}
Stefan Kratsch.
\newblock On polynomial kernels for sparse integer linear programs.
\newblock {\em Journal of Computer and System Sciences}, 82(5):758--766, 2016.

\bibitem{lenstra1990approximation}
Jan~Karel Lenstra, David~B Shmoys, and {\'E}va Tardos.
\newblock Approximation algorithms for scheduling unrelated parallel machines.
\newblock {\em Mathematical programming}, 46(1-3):259--271, 1990.

\bibitem{mnich2015scheduling}
Matthias Mnich and Andreas Wiese.
\newblock Scheduling and fixed-parameter tractability.
\newblock {\em Mathematical Programming}, 154(1-2):533--562, 2015.

\bibitem{onn2010nonlinear}
Shmuel Onn.
\newblock Nonlinear discrete optimization.
\newblock {\em Zurich Lectures in Advanced Mathematics, European Mathematical
  Society}, 2010.

\bibitem{sahni1976algorithms}
Sartaj~K Sahni.
\newblock Algorithms for scheduling independent tasks.
\newblock {\em Journal of the ACM (JACM)}, 23(1):116--127, 1976.

\bibitem{xu2013exact}
Liang Xu, Zhou Xu, and Dongsheng Xu.
\newblock Exact and approximation algorithms for the min--max k-traveling
  salesmen problem on a tree.
\newblock {\em European Journal of Operational Research}, 227(2):284--292,
  2013.

\end{thebibliography}

\end{document}